\documentclass{llncs}
\usepackage[utf8]{inputenc}

\usepackage{amssymb,amsmath,bm}
\usepackage{graphicx}
\usepackage{fancyhdr}
\usepackage[linesnumbered,ruled,vlined]{algorithm2e}
\SetKwInput{KwInput}{Input}                
\SetKwInput{KwOutput}{Output}              

\RequirePackage{fix-cm}
\newcommand\norm[1]{\left\lVert#1\right\rVert}
\DeclareMathSizes{10}{10}{6}{4}
\title{Faster Lattice Enumeration}
\author{Mithilesh Kumar}
\institute{Simula UiB, Bergen, Norway}

\begin{document}

\maketitle

\begin{abstract}
A lattice reduction is an algorithm that transforms the given basis of the lattice to another lattice basis such that problems like finding a shortest vector and closest vector become easier to solve. Some of the famous lattice reduction algorithms are LLL and BKZ reductions. We define a class of bases called \emph{obtuse bases} and show that any lattice basis can be transformed to an obtuse basis in $\mathcal{O}(n^4)$ time. A shortest vector $\bm{s}$ can be written as $v_1\bm{b}_1+\cdots+v_n\bm{b}_n$ where $\bm{b}_1,\dots,\bm{b}_n$ are the input basis vectors and $v_1,\dots,v_n$ are integers. When the input basis is obtuse, all these integers can be chosen to be positive for a shortest vector. This property of the obtuse basis makes lattice enumeration algorithm for finding a shortest vector exponentially faster. Moreover, \emph{extreme pruning}, the current fastest algorithm for lattice enumeration, can be run on an obtuse basis.
\end{abstract}
\section{Introduction}
A lattice is specified by integer linear combinations of a linearly independent basis $\mathcal{B}=\{\bm{b}_1,\dots,\bm{b}_n\}$, i.e. $\Lambda=\{v_1\bm{b}_1+\dots+v_n\bm{b}_n\mid v_i\in \mathbb{Z}\}$. Given only $\mathcal{B}$, finding a shortest vector in the lattice is \textsf{NP}-hard under randomized reductions \cite{Ajtai}, implying that a polynomial-time algorithm for the same is unlikely. Because of the average-case hardness of the problem, it has become significant in cryptography. With no known quantum attacks, lattice based cryptography is prime candidate for post-quantum cryptography. Lattices have been studied in mathematics much before any cryptographic application became widespread and researchers have discovered many ways to find a shortest vector when only a basis for the lattice is given.
Well-known approach to solve this problem are \emph{lattice reduction}, \emph{sieving}, and \emph{lattice enumeration}. In lattice reduction, a unimodular transformation is applied to the lattice basis. This makes the basis vectors shorter and more orthogonal (angles between basis vectors are close to being $\pi/2$). Usually reduction algorithms run in polynomial-time and output a vector which is not necessarily a shortest vector. The most widely used lattice reduction algorithms are LLL \cite{lll} and BKZ \cite{bkz}. 

Enumeration and sieving algorithms try to find a shortest vector, but require an exponential running time \cite{sieve1,sieve2,disc1,disc2}. 
An enumeration algorithm tries to find a shortest vector $\bm{s}=\sum_{i=1}^nv_i\bm{b}_i$ by finding the unknown integer coefficients $v_i\in \mathbb{Z}$ by listing all lattice vectors in some sphere of radius $R$. 
This makes the running time super-exponential in the size of the input. In addition, the running time is highly sensitive to the quality of the lattice basis provided. Typically, the shorter and more orthogonal the lattice basis is, the faster will the enumeration be. Usually, a lattice reduction is used to preprocess the lattice basis. 
An enumeration algorithm can be described by a search tree in which each node corresponds to a lattice vector (not necessarily within the chosen sphere of radius $R$). This suggests another way to speed-up enumeration by throwing away parts of the search tree where the likelihood of finding a shortest vector is small. This approach is called \emph{pruning} \cite{SE94,SH95,GNR10}. Pruning speeds up enumeration at the cost of some probability of failure to return the shortest vector. Note that even after pruning, the running time of enumeration is exponential. Although lattice enumeration is nearly a brute-force search algorithm, it continues to remain a practical success.

\textbf{Our contribution:}
A path from the root node to a leaf in the seach tree representing an enumeration algorithm can been as an assignment for unknown coefficients $v_1,\dots, v_n$ in $\bm{s}=v_1\bm{b}_1+\cdots+v_n\bm{b}_n$. At each internal node, depending on variables already fixed, the algorithm computes an interval $I_i$ for some unknown variable $v_i$ and branches on each integer in $I_i$. Clearly the shorter the interval $I_i$, the lesser will the number of nodes in the search tree, and hence the faster will the enumeration be.

Recently, we\footnote{author names withheld for anonymity as the accepted paper is not in public domain yet} studied the effect of permutation of the basis on lattice enumeration, and proposed a heuristic algorithm to predict signs of coefficients $v_i$ for a possible shortest vector. Given these signs, the interval $I_i$ can be made shorter by keeping only those integers that have the predicted sign of $v_i$. 

In this paper, we provide a deterministic algorithm that in $\mathcal{O}(n^4)$ time decides whether signs of coefficients of the basis can be obtained. If this is not possible, the algorithm outputs another basis for which it would be possible to determine the signs. This algorithm needs to be run once before an enumeration algorithm. We show in Section \ref{enum} that there are roughly $\sqrt{n-1}/\sqrt{2\pi}R$ (where $R$ is the radius of the sphere in which the enumeration algorithm searches for a shortest vector) fraction of nodes at which the enumeration tries impossible values. Since this search tree has expenentially many nodes, the speed-up will be exponential. Moreover, \emph{extreme pruning} can still be applied in the enumeration algorithm. 

In Section \ref{prelims} we define common terminology and set notation. In Section \ref{red}, we describe the new reduction algorithm and in Section \ref{bound} we discuss how to bound the length of basis vectors returned by the algorithm in Section \ref{red}. We apply the reduction algorithm in the enumeration algorithm in Section \ref{enum}.
\section{Preliminaries}\label{prelims}
We denote vectors by boldface letters and scalars as normal letters. Sets are represented by capital letters. For inclusion or exlusion of a single element into or out of a set is denoted by $+$ or $-$ sign, i.e. $A\cup \{a\}$ is written as $A+a$ and $A\setminus \{a\}$ is written as $A-a$. 
A lattice is specified by a set of $n$ linear independent vectors in $\mathbb{R}^n$. Let $\mathcal{B}=\{\bm{b}_1,\dots,\bm{b}_{n}\}$ be a lattice basis. Then, the set of vectors in the lattice $\Lambda$ are integer linear combinations of $\mathcal{B}$. $$\Lambda=\{v_1\bm{b}_1+\cdots+v_{n}\bm{b}_{n}\mid \forall i\in [n] ~ v_i\in \mathbb{Z}\}$$
The inner product is the sum of the point-wise product of corresponding co-ordinates, i.e. $\bm{a}\cdot\bm{b}=\sum_{i=1}^na_ib_i$. We are interested in the Euclidean norm which is defined as $\norm{\bm{b}}=\sqrt{\bm{b}\cdot\bm{b}}$. 

A graph $G=(V,E)$ is specified by a set $V$ of points called vertices which are connected by a set $E\subseteq V\times V$ of edges. A clique or complete graph on $n$ vertices is a graph in which every pair $u,v\in V$ is connected via an edge $uv\in E$.
\section{New lattice reduction}\label{red}
In this section, we define a special class of basis vectors called \emph{obtuse bases} and show that it is possible to transform any basis to an obtuse basis in polynomial-time.
\begin{definition}[Obtuse basis]
	Let $\mathcal{B}:=\{\bm{b}_1,\dots,\bm{b}_n\}$ be a basis for lattice $\Lambda$. The basis $\mathcal{B}$ is called \emph{obtuse} if for all $\bm{b}_i\neq\bm{b}_j\in \mathcal{B}$, we have that $\bm{b}_i\cdot\bm{b}_j\leq 0$.
\end{definition}
We start with proving a simple yet powerful relation between a shortest vector and an obtuse basis.
	\begin{lemma}\label{samesign}
		Let $\mathcal{B}$ be an obtuse basis of a lattice $\Lambda$ and $\bm{s} = \sum_{i=1}^n v_i\bm{b}_i$ be a shortest vector where $\forall i\in [n]~ v_i\in \mathbb{Z}$. Then, $\forall i\in [n]~v_i\geq 0$ or $\forall i\in [n]~v_i\leq 0$.
	\end{lemma}
\begin{proof}
Computing $\bm{s}\cdot\bm{s}$, we get $\norm{\bm{s}}^2 = \sum_{i = 1}^n v_i^2\norm{\bm{b}_i}^2 + \sum_{i\neq j} v_i v_j \bm{b}_i\cdot \bm{b}_j$. Since for all $i\neq j \quad\bm{b}_i\cdot \bm{b}_j < 0$, the above sum is the smallest possible when for all $i\neq j\quad v_i v_j \geq 0$. This implies that either $\forall i\in [n]~v_i\geq 0$ or $\forall i\in [n]~v_i\leq 0$. 
\end{proof}


However the original basis is not necessarily obtuse. 
A natural question is 'Can we transform the basis such that the new basis is obtuse?' A first attempt to get such a transformation would be to flip the sign of a vector i.e. $\bm{b}\to -\bm{b}$. If for any pair $\bm{b}_i,\bm{b}_j\in \mathcal{B}$ we get $\bm{b}_i\cdot\bm{b}_j>0$, then we may try replacing $\bm{b}_i$ by $-\bm{b}_i$ and check whether this makes the basis obtuse. Does there exist an algorithm that uses only sign-flip to make a given basis obtuse?

To investigate this, let us first construct the following \textbf{sign graph} $G:=(V,E)$ : For each basis vector $\bm{b}_i$, there is a corresponding node $a_i\in V$. For each pair of vertices $a_i, a_j$, the weight of edge $a_ia_j$ is defined as $wt(a_ia_j):=\bm{b}_i\cdot\bm{b}_j$. An edge is called \emph{positive} if its weight is a positive number and called \emph{negative} if its weight is a negative number. We color positive edges blue and negative edges red. Using this graph, we will find a sequence of transformations that make the basis obtuse. In further discussions, we identify vertices with basis vectors and vice versa.

Consider any algorithm that can at any step only flip the sign of any vector. We call it a \emph{sign-flip} algorithm.
In the next lemma we show that such an algorithm does not always succeed in making the basis obtuse.	
	\begin{lemma}\label{odd}
	Let $n$ be odd and ${n\choose 2} \equiv 1 \mod 4$. Let $x$ be the number of negative edges and $y$ be the number of positive edges. If $x-y\equiv -1 \mod 4$, then a sign-flip algorithm can not make all edges negative.
	\end{lemma}
\begin{proof}
Define $S=x-y$ for the pair $(x,y)$. Consider a step that flips the sign of a basis vector which has $d^+$ many positive edges and $d^-$ many negative edges incident on it. After the step, the number of negative edges will be $x'=x-d^-+d^+$ and the number of positive edges will be $y'=y-d^++d^-$. Then, $S'=x'-y'=x-y+2(d^+-d^-)$. Since, $d^++d^-=n-1\equiv 0 \mod 2$, we have $S'\equiv S\equiv -1 \mod 4$. The algorithm terminates when $S'={n\choose 2}$, i.e. $S'\equiv 1\mod 4$ which is impossible. Hence, the algorithm can not succeed.
\end{proof}
\begin{figure}[t]\label{fig:cliquepartition}
	\centering
	\includegraphics[width=0.3\textwidth]{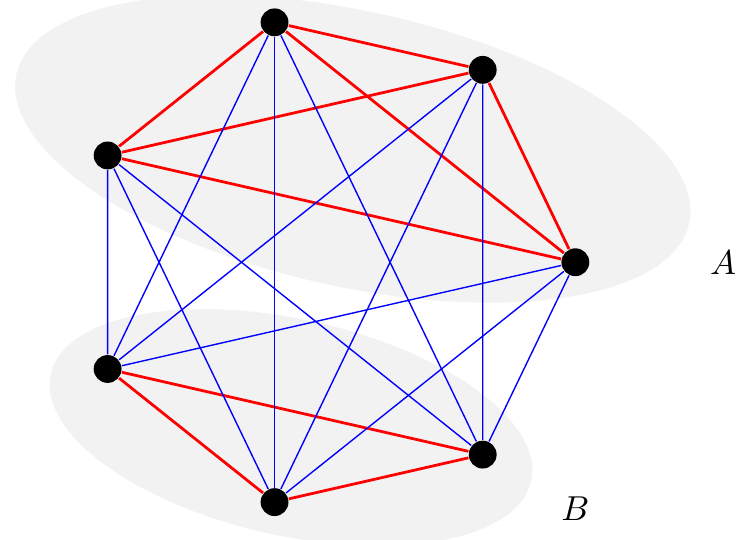}
	\caption{Complete graph $G=(V,E)$ in which the vertex set $V$ has been partitioned into $(A,B)$ such that the edges in graphs induced on $A$ and $B$ denoted by $G[A]$ and $G[B]$ contain only red edges. Every other edge is blue.}
\end{figure}
In the following lemma, we show a necessary and sufficient condition for a sign-flip algorithm to succeed.
	\begin{lemma}\label{clique}
		A sign-flip algorithm succeeds in making a basis obtuse if and only if the vertex set can be partitioned into two parts $(A,B)$ such that $G[A]$ and $G[B]$ have only negative edges and rest of the edges are positive.
	\end{lemma}
\begin{proof}
Let $(A,B)$ be such a partition. See Figure \ref{fig:cliquepartition}. Red edges represent negative edges and blue edges represent positive edges. In this case we show that the algorithm will succeed. Let $a\in A$ and we flip its sign. Since, all edges between $a$ and vertices in $B$ were positive, they become negative after the flip. The clique $G[B+a]$ will have only negative edges. The edges between $a$ and vertices in $A$ were negative and will become positive after the flip. Hence, $a$ can be \emph{sent} to $B$ such that $|A|$ decreases by $1$ and the new partition $(A-a,B+a)$ still satisfies the property in the lemma. Proceeding in this way, every vertex from $A$ can be sent to $B$ implying the algorithm terminates successfully.

For the other direction, suppose now that the algorithm succeeds. At any intermediate stage of the algorithm, let us define $B$ to be the set of vertices whose signs will not be flipped again. Define the rest of the vertices to belong to $A$. We show that $(A,B)$ is a partition of the vertex set that satisfies the conditions in the lemma. Clearly, $G[B]$ must be a clique having only negative edges. \\
\textbf{Case 1}: There exists $a\in A$ such that there exist vertices $w,y\in B$ such that $aw$ and $ay$ have opposite sign. In this case, the signs of these edges will continue to remain opposite and hence the algorithm can not succeed contradicting the assumption. \\
The above case implies that edges from a vertex in $A$ to vertices in $B$ have the same color/sign. In case, these edges are red/negative, we send them into $B$ maintaining that $G[B]$ is a clique on red edges. This leaves with the situation that there are only positive edges between vertices of $A$ and vertices of $B$.\\
\textbf{Case 2}: There exists $w,y\in A$ such that $wy$ is positive. Then for any vertex $a\in B$ we have a triangle $awy$ of positive edges. We can not make this triangle red by sign flip. Hence, this case can not exist under the assumption that the algorithm succeeds. This concludes the proof of the lemma. 
\end{proof}
We can prove Lemma \ref{odd} using Lemma \ref{clique}. Suppose $V(G)$ could be partitioned into $(A,B)$ as required by Lemma \ref{clique}. Say $|A|=n_1$ and $|B|=n_2$. Then $n_1+n_2=n$ and at least one of $n_1$ and $n_2$ is even since $n$ is odd in Lemma \ref{odd}. Moreover, the number of negative edges (contained only in $G[A]$ and $G[B]$) are $x={n_1\choose 2}+{n_2\choose 2}={n_1+n_2\choose 2}-n_1n_2={n\choose 2}-n_1n_2$ and the number of positive edges (edges between $A$ and $B$) are $y=n_1n_2$. Their difference $x-y={n\choose 2}-2n_1n_2$. Hence $x-y\equiv {n\choose 2}~\mod 4\equiv 1~\mod 4$. Since Lemma \ref{odd} is provided with $x-y\equiv -1~\mod 4$, a sign-flip algorithm can not succeed.

Lemma \ref{clique} can be converted into the Sign-flip algorithm \ref{signflip} described below.

\begin{algorithm}[H]\label{signflip}
\DontPrintSemicolon
\KwInput{A basis $\mathcal{B}=\{\bm{b}_1,\dots,\bm{b}_n\}$}
\KwOutput{An obtuse basis $\mathcal{B}_o=\{\bm{b}'_1,\dots,\bm{b}'_n\}$}

Construct sign-graph $G$.\;
Pick any vertex $u\in V(G)$\;
Define $A\subseteq V(G)$ such that $u$ has only red edges to vertices in $A$, i.e. $A$ is the set of red neighbors of $u$.\;
Define $B\subseteq V(G)$ such that $u$ has only blue edges to vertices in $B$, i.e. $B$ is the set of blue neighbors of $u$\;
\eIf{any of $G[A]$ or $G[B]$ contain a blue edge} {Output FAIL.}{Flip the sign of every vector in $B$, i.e. $\forall~\bm{b}_i\in B~\bm{b}'_i:=-\bm{b}_i$}

Return basis $\mathcal{B}_o=\{\bm{b}'_1,\dots,\bm{b}'_n\}$
\caption{Sign-flip}
\end{algorithm}

In case, the Algorithm \ref{signflip} outputs FAIL, we would require to run some other algorithm. Now we proceed to develop an algorithm that succeeds to obtain an obtuse basis in all cases. Our task is to find a basis $\{\bm{b}'_i\}$ such that for all $i\neq j$ we have that $\bm{b}'_i\cdot\bm{b}'_j\leq 0$.
\paragraph{Description of Algorithm \ref{algo}:} Algorithm \ref{algo} starts with a sub-lattice (possibly of dimension $1$) which is obtuse (the clique $K$). It picks a basis vector in the complement of the current obtuse basis and changes it such that the resulting basis vector along with the current obtuse sub-lattice form a bigger obtuse sub-lattice. Once a basis vector is included in the obtuse sub-lattice, it is not modified again. This way at each iteration in the while loop, the dimension of the obtuse sub-lattice increases by one. 

In step 5, the algorithm expresses $\bm{b}'_j$ as a linear combination of basis vectors in $\{\bm{b}_j\}\cup K$ where the coefficient of $\bm{b}_j$ is fixed to be non-zero integer. Steps 6,7,8 are there to find the unknown coefficients $a_i$. In step 6, the system of linear inequalities are expressing the requirement that the new basis vector $\bm{b}'_j$ is obtuse to all basis vectors in $K$. Later we'll show that the matrix $A$ is invertible and hence $\bm{x}$ in step 7 is well defined. The unknown coefficients $\bm{v}$ are obtained by taking floor of entries in $\bm{x}$. It remains to be shown that this choice of $\bm{v}$ satisfies the system of linear inequalities in step 6. Assuming that it does, the algorithm extends $K$ to include the basis vector $\bm{b}'_j$ and the goes back to step 3. Finally, the algorithm outputs the required obtuse basis.

\begin{algorithm}[H]\label{algo}
\DontPrintSemicolon
\KwInput{A basis $\mathcal{B}=\{\bm{b}_1,\dots,\bm{b}_n\}$}
\KwOutput{An obtuse basis $\mathcal{B}_o=\{\bm{b}'_1,\dots,\bm{b}'_n\}$}

Reorder basis $\mathcal{B}$ such that the basis vectors in any maximal red clique $K$ in $G$ are at the end.\;
Initialize $\mathcal{B}_o$ by the basis vectors from $K$.\;
\While{$|\mathcal{B}_o|< n$}
{
	Let $\bm{b}_j\in G-K$ be the vector with largest number of red edges to $K$.\; 
	Define $\bm{b}'_j:=y_j\bm{b}_j+\sum_{\bm{b}_i\in K}v_i\bm{b}_i$ such that 	$y_j,v_i\in \mathbb{Z}$ and $y_j\neq 0$.\;
	Construct the system of linear inequalities defined by $$\forall~ \bm{b}_\ell\in K\quad \bm{b}'_j\cdot \bm{b}_\ell\leq 0$$ Equivalently, $$\forall~ \bm{b}_\ell\in K\quad \sum_{\bm{b}_i\in K}v_i\bm{b}_i\cdot \bm{b}_\ell\leq -y_j\bm{b}_j\cdot\bm{b}_\ell$$ which is of the form
$$A\bm{v}\leq y_j\bm{c}$$\;
	Compute $\bm{x}=y_jA^{-1}\bm{c}$.\;
   Take $\bm{v}$ to be the floor of the entries in $\bm{x}$.\;
	Choose $y_j\neq 0$ such that $\bm{b}'_j$ is as orthogonal to each $\bm{b}_i\in K$. This can be done by trying $y_j\in [-n,n]$\;
	Include $\bm{b}'_j$ into $\mathcal{B}_o$ and the corresponding vertex into the cliques $K$.
}	
Return basis $\mathcal{B}_o=\{\bm{b}'_1,\dots,\bm{b}'_n\}$
\caption{Compute obtuse basis}
\end{algorithm}

\paragraph{Correctness:} The correctness of the algorithm follows if we prove that $A$ is invertible and $\bm{a}$ obtained in step 8 satisfies the inequalities in step 6. Let us first show that $A$ is invertible. Consider the matrix $L = U^TU$. In this matrix $L_{ij} =  \bm{b}_i\cdot \bm{b}_j $. The matrix $A$ defined from the terms $\bm{b}_k\cdot \bm{b}_j$ in the above system is a sub-matrix of $L$. $A$ can also be seen as $V^TV$ where the rows/columns of $V$ are the basis vectors $\bm{b}_k$ for $i\leq k\leq n$. Since the rank of $V$ and $V^T$ is $n-i$, the rank of $A$ is $n-i$ and hence invertible. 

For step 8, we show that $A$ is an M-matrix.
\begin{definition}[M-matrix]
Let $A$ be a $n\times n$ real matrix such that $a_{ij}\leq 0$ for all $i\neq j, 1\leq i,j\leq n$ and $A$ is positive definite. Then $A$ is called an M-matrix.
\end{definition}
There are many properties and equivalent definitions of $M$-matrices \cite{PLEMMONS}. If $A$ is an $M$-matrix, then $A\bm{z}\geq 0$ implies that $\bm{z}\geq 0$ (i.e. every entry in $\bm{z}$ is non-negative). If $A$ is an $M$-matrix, then $A^{-1}\geq \bm{0}$ where $\bm{0}$ is zero matrix.


At any step, the matrix $A$ is symmetric with diagonal entries positive and the rest of the entries negative. In fact, $A$ is positive definite. This is because $A=V^TV$ and for any $\bm{0}\neq \bm{z}\in R^{n-i}$, we have that $\bm{z}^TA\bm{z}=\bm{z}^TV^TV\bm{z}=(V\bm{z})^T(V\bm{z})=\norm{V\bm{z}}^2>0$ as $V\bm{z}\neq \bm{0}$. This implies that all eigenvalues of $A$ are positive and hence $A$ is an $\bm{M}$-matrix. 
\begin{theorem}
Let $A$ be an M-matrix. Then, if $A\bm{v}\leq \bm{b}$, then $\bm{v}\leq A^{-1}\bm{b}$. In addition, if $\bm{u}\leq \bm{v}$, then $A\bm{u}\leq \bm{b}$.
\end{theorem}
\begin{proof}
Since $A$ is an $M$-matrix, then $A\bm{z}\geq 0$ implies that $\bm{z}\geq 0$. Let $\bm{x}=A^{-1}\bm{b}$. Since $0\leq \bm{b}-A\bm{v}=A(\bm{x}-\bm{v}) $, we get that $\bm{x}-\bm{v}\geq 0$. Hence, $\bm{v}\leq A^{-1}\bm{b}$.

Let $\bm{u}\leq \bm{v}$. Using $\bm{b}-A\bm{v}\geq 0$ and $A^{-1}\geq 0$, we get that $A^{-1}\bm{b}-\bm{v}\geq 0$. This implies that $A^{-1}-\bm{u}\geq 0$ and hence $A\bm{u}\leq \bm{b}$.
\end{proof}
Hence, writing the above system of equations as $A\bm{v}\leq \bm{c}$, we get rational solutions $\bm{v}\leq A^{-1}\bm{c}$. This gives us an immediate solution $\lfloor A^{-1}\bm{c}\rfloor$. 

\paragraph{Running Time:} Each iteration of the while loop takes $\mathcal{O}(n^3)$ time. There are at most $n$ iterations in the while loop. Hence, the running time of the algorithm is bounded by $\mathcal{O}(n^4)$.
\section{Enumeration algorithm}\label{enum}
A shortest vector $\bm{s}$ in a given lattice $\Lambda$ can be obtained by enumerating all lattice vectors in a ball of some radius $R$. This seemingly brute-force search algorithm happens to be one of the fastest known algorithms for computing a shortest vector in a general lattice. Let us briefly describe the enumeration algorithm. Suppose $\bm{s}$ can be expressed as $\bm{s}=\sum_{i=0}^{n-1}v_i\bm{b}_i$ where $\bm{b}_i$ are basis vectors of $\Lambda$ and $v_i\in \mathbb{Z}$. The algorithm intends to list all vectors satisfying $\norm{\bm{s}}\leq R$. Using the Gram-Schmidt orthogonal basis $\mathcal{B}^*=\{\bm{b}^*_0,\dots,\bm{b}^*_{n-1}\}$, where $\bm{b}^*_i=\bm{b}_i-\sum_{j=0}^{i-1} \mu_{i,j}\bm{b}^*_{j}$ where $\mu_{i,j}=\frac{\bm{b}_i\cdot\bm{b}^*_j}{\bm{b}^*_j\cdot\bm{b}^*_j}$, $\bm{s}$ can be rewritten as $$\bm{s}=\sum_{i=0}^{n-1}\big (v_i+\sum_{j=i+1}^{n-1}\mu_{i,j}v_j\big )\bm{b}^*_i$$
Using that $\bm{b}^*_i\cdot\bm{b}^*_j=0$, we obtain the following inequality equivalent to $\norm{\bm{s}}\leq R$, $$\sum_{i=0}^{n-1}\bigg\vert v_i+\sum_{j=i+1}^{n-1}\mu_{i,j}v_j\bigg\vert^2\norm{\bm{b}^*_i}^2\leq R^2$$ Since each term in the above sum is non-negative, each term must be bounded by $R^2$. This observation leads to a search tree.
Levels in the search tree are numbered from $n-1$ down to $0$. Variable $v_{n-1}$ is evaluated at level $n-1$ and $v_{0}$ is evaluated at level $0$. The following recurrence relation is used to compute the range of allowed values for variable $v_{level}$:
\begin{equation}\label{eq:range}
\bigg\vert v_{level}+\sum_{k=level+1}^{n-1}\mu_{\tiny{level,k}}v_k\bigg\vert\leq \frac{\sqrt{R^2-\sum_{k=level+1}^{n-1}|v_k+\sum_{\ell=k+1}^{n-1}\mu_{k,\ell}v_\ell|^2\norm{\bm{b}^*_k}^2}}{\norm{\bm{b}^*_{level}}}
\end{equation}
Notice that the range for $v_{level}$ could contain both positive and negative integers. If we know \emph{a priori} that the input basis $\mathcal{B}$ is obtuse, we could simply discard the negative integers in the range for $v_{level}$ thereby cutting down the search tree significantly. Given a lattice basis $\mathcal{B}$, it is easy to check whether $\mathcal{B}$ is obtuse. In case $\mathcal{B}$ is obtuse, we choose only non-negative integers in the enumeration. If $\mathcal{B}$ is not obtuse, we check whether the Sign-flip algorithm \ref{signflip} can make the basis obtuse. If this fails as well, we run Algorithm \ref{algo} to obtain a new basis $\mathcal{B}_o$ which is obtuse.

Next we show that this leads to an exponential speed-up in the enumeration algorithm. At any node, if the computed range contains only negative integers, then we can recurse back immediately. This will happen whenever 
\begin{equation}\label{eq:negatives}
 \sum_{k=level+1}^{n-1}\mu_{\tiny{level,k}}v_k > \frac{\sqrt{R^2-\sum_{k=level+1}^{n-1}|v_k+\sum_{\ell=k+1}^{n-1}\mu_{k,\ell}v_\ell|^2\norm{\bm{b}^*_k}^2}}{\norm{\bm{b}^*_{level}}}
\end{equation}

If equation \ref{eq:negatives} is satisfied by exponentially many values of $\{v_{level},\dots,v_{n-1}\}$, we get exponential speed-up with this observation itself. Instead of estimating the number of such values, we estimate the fraction of nodes in the search tree for which the ranges contain both positive and negative integers. We call nodes where both positive and negative integers are in the range as \emph{good nodes}.

Variable $v_{level}$ can have both negative and positive values only when 
\begin{equation}\label{eq:bothsigns}
\bigg\vert \sum_{k=level+1}^{n-1}\mu_{\tiny{level,k}}v_k\bigg\vert < \frac{\sqrt{R^2-\sum_{k=level+1}^{n-1}|v_k+\sum_{\ell=k+1}^{n-1}\mu_{k,\ell}v_\ell|^2\norm{\bm{b}^*_k}^2}}{\norm{\bm{b}^*_{level}}}
\end{equation}
Note that Equation \ref{eq:bothsigns} is simply Equation \ref{eq:range} with $v_{level}=0$. 

To find the number of \emph{good nodes}, we need to find the number of assignments of $\{v_{level-1},\dots v_{n-1}\}$ satisfying Equation \ref{eq:bothsigns}. Equivalently, we need to find the number of integer points in the volume enclosed by Equation \ref{eq:range} with $v_{level}=0$.

Rearranging Equation \ref{eq:range}:
\begin{equation}
\sum_{k=level}^{n-1}|v_k+\sum_{\ell=k+1}^{n-1}\mu_{\ell,k}v_\ell|^2\norm{\bm{b}^*_k}^2\leq R^2
\end{equation}
Define $u_k:=v_k+\sum_{\ell=k+1}^{n-1}\mu_{k,\ell}v_\ell$ and $r_k:=\norm{\bm{b}^*_k}$. This gives us the following form
\begin{equation}\label{eq:ellipsoid}
\sum_{k=level}^{n-1}|u_k|^2r_k^2\leq R^2
\end{equation}
Equation \ref{eq:ellipsoid} is an equation of an ellipsoid in $n-level$ dimensions in variables $u_k$. With 
$$T:=\begin{bmatrix}
1 & \mu_{level,level+1} &  \mu_{level,level+2}& \dots & \mu_{level,n-1}\\
0    & 1 &\mu_{level+1,level+2} & \dots & \mu_{level+1,n-1}\\
\vdots & \vdots &\vdots &\vdots &\vdots\\
0 & \dots &\dots &\dots & 1
\end{bmatrix}
$$
we can write $\bm{u}=T\bm{v}$. Note that $T$ is an orthogonal Gram-Schmidt transform. Next consider the following sphere 
\begin{equation}\label{eq:sphere}
\sum_{k=level}^{n-1} |w_k|^2\leq R^2
\end{equation}
Equation \ref{eq:sphere} represents a sphere of dimension $n-level$. With $$S:=\begin{bmatrix}
r_{level} & 0 & \dots & 0\\
0    & r_{level+1} & \dots & 0\\
\vdots & \vdots &\vdots &\vdots\\
0 & \dots &\dots & r_{n-1}
\end{bmatrix}$$
we can write $\bm{w}=S\bm{u}$. Chaining these two linear transformations we get
$$\bm{w}=ST\bm{v}$$
Substituting $w_{level}=0$ into Equation \ref{eq:sphere} gives us a sphere in dimension $n-level-1$. Now, the volume of an $n-level-1$-dimensional ball is $$V(n-level-1,R)=\frac{\pi^{\frac{n-level-1}{2}}}{\Gamma(\frac{n-level-1}{2}+1)}R^{n-level-1}$$ 
The volume of the region enclosed by Equation \ref{eq:bothsigns} is given by $$V_{n-level-1}=V(n-level-1,R)det((ST)^{-1}) = \frac{\pi^{\frac{n-level-1}{2}}}{\Gamma(\frac{n-level-1}{2}+1)}\frac{R^{n-level-1}}{\Pi_{k=level}^{n-1}\norm{b}^*_k}$$
Notice that the volume of region enclosed by Equation \ref{eq:range} is given by $$V_{n-level}=V(n-level,R)det((ST)^{-1}) = \frac{\pi^{\frac{n-level}{2}}}{\Gamma(\frac{n-level}{2}+1)}\frac{R^{n-level}}{\Pi_{k=level}^{n-1}\norm{b}^*_k}$$
Hence the fraction of good nodes is $$\frac{V_{n-level-1}}{V_{n-level}}=\frac{\Gamma(\frac{n-level}{2}+1)}{\Gamma(\frac{n-level-1}{2}+1)}\frac{1}{\sqrt{\pi}R}\approx\frac{\sqrt{\frac{n-level-1}{2}}}{\sqrt{\pi}R}=\frac{\sqrt{n-level-1}}{\sqrt{2\pi}R}$$

\begin{theorem}\label{thm:main}
Given any basis for a lattice in dimension $n$, the fraction of nodes at which both negative and positive values are evaluated is roughly $\sqrt{n-1}/\sqrt{2\pi}R$ where $R$ is the radius of the sphere for which the enumeration algorithm is run.
\end{theorem}
Notice that in the search tree, whenever a node with negative value is encountered, the enumeration algorithm discards the entire sub-tree rooted at this node. A simple example would a case in which every interval in the enumaration is $\{-1,0,1\}$. The size of such a tree is $3^n$. If we discard all nodes with value $-1$, the tree size reduces to $2^n$, an exponential decrease in the size of the seach tree that can not be expressed as $c3^n$ for any constant $c$. 

To get some insight, we assume that all intervals have the same length $I$ at each node is the enumeration. This means that there are $I^n$ nodes in the search tree. In the lattice enumeration, at the nodes of type listed by Theorem \ref{thm:main}, we reduce the size of the interval by half on average: for symmetric intervals $[-a,a]$ we discard $[-a,-1]$, for skewed intervals like $[-a-r,a-r]$ and $[-a+r,a+r]$ we discard on an average $a$-many values. Theorem \ref{thm:main} implies that there are approximately $1/2\sqrt{2\pi}R$ fraction of nodes at each level in the search tree at which we get to decrease the interval size by half. To bound the number of nodes, we can see this as removing $1/2\sqrt{2\pi}R$ fraction of intervals at half of nodes of the search tree or cutting down all intervals by a factor $1/4\sqrt{2\pi}R$. Hence, the effective interval size is $\alpha I$ where $\alpha=1-1/4\sqrt{2\pi}R$. Hence, the new tree size is $\alpha^nI^n$ which is an exponential decrease in the size of the tree.

The above paragraph indicates that our algorithm will gain expenential-speed up over standard enumeration. Notice that extreme pruning can be applied on obtuse basis as well.


\section{Discussion}\label{bound}
The coefficients $\bm{v}$ obtained in Algorithm \ref{algo} are not necessarily optimal. 
It would be nice if we have some bound on how long the modified vectors could be.
Let us consider the problem geometrically. Given basis vectors $\mathcal{B}_d=\{\bm{b}_1,\dots,\bm{b}_d\}$, the task is to find the set of vectors that are obtuse to $\mathcal{B}_d-\bm{b}_d$ and in the linear expansion of vectors in terms of $\mathcal{B}_d$, the coefficient of $\bm{b}_d$ is non-zero. This last condition is needed to make sure that the resultant vector is linearly independent of $\{\bm{b}_1,\dots,\bm{b}_{d-1}\}$.

For each vector $\bm{b}_i\in \mathcal{B}_d-\bm{b}_d$, let $P_i$ denote the halfspace $\bm{b}_i\cdot \bm{x}\leq 0$. The region of interest for us is $C=\bigcap_i P_i$. Notice that $C$ also contains the span of $\mathcal{B}_d-\bm{b}_d$. It may happen that the length bound we get is likely due to a vector in the span of $\mathcal{B}_d-\bm{b}_d$. To compensate for this, we can increase the length bound by a factor of 2. As $C$ is intersection of convex objects, it is convex. Each pairwise intersection $P_i\cap P_j$ gives an edge of $C$ passing through the origin, i.e. we have ${d-1\choose 2}$ such edges. Consider the intersection of a sphere $S$ of radius $r$ centered at the origin with $C$. A bound on the radius $r$ such that the volume of this intersection is equal to the volume of the fundamental parallelepiped of $\Lambda$ gives us a bound on the length of a lattice vector which is obtuse to $\mathcal{B}_d-\bm{b}_d$. In case, $\mathcal{B}_d$ is orthogonal, the volume of $S\cap C$ is $\frac{Vol(S)}{2^{d}}$(because there are $2^d$ equal quadrants).
It is reasonable to assume that in general volume of $S\cap C$ will be proportional to $$\frac{Vol(S)}{2^d}\frac{\Omega}{\Omega_0}$$ where $\Omega$ is the solid angle that the vertex of $C$ at the origin makes and $\Omega_0$ is the solid angle of a quadrant. Equating this to the volume of the fundamental parallelepiped gives a bound on $r$. We leave a rigorous bound for the length of basis vectors for future work.

Future work includes studying how \emph{obtuse} LLL-reduced and BKZ-reduced basis is? There is likely a better algorithm to make a basis obtuse with a provable guarantee in the quality (in terms of basis vector lengths and orthogonality). What happens when BKZ and Algorithm \ref{algo} are  interleaved with multiple iterations? An implementation of this approach is needed to access its impact. Whether the notion of obtuse basis is relevant for sieving algorithms needs to be investigated as well.



\bibliography{references}
\bibliographystyle{splncs04}

\end{document}